\documentclass[journal,draftcls,onecolumn,12pt,twoside]{IEEEtran}
\usepackage{algorithm} 
\usepackage{algpseudocode}
\usepackage{graphicx}
\usepackage{amssymb}
\usepackage{amsfonts}
\usepackage{amsmath}
\usepackage{float}
\usepackage{url}
\usepackage{bm}
 \usepackage[colorlinks=true]{hyperref}
\hypersetup{urlcolor=blue, citecolor=blue}

\IEEEoverridecommandlockouts
\newtheorem{Theorem}{Theorem}
\newtheorem{Proposition}{Proposition}

\newtheorem{Definition}{Definition}

\newtheorem{lemma}{Lemma}

\newenvironment{proof}{\textit{Proof\,:}} { $\square$}

\ifCLASSINFOpdf
\else
\fi
\begin{document}
\title{Analytical lower bounds for the size of elementary trapping sets of variable-regular  LDPC codes with any girth and irregular ones with  girth 8}
\author{ Farzane Amirzade and Mohammad-Reza~Sadeghi\\
\thanks{%
Manuscript received May ??, ????; revised November ??, ????.}
\thanks{  M.-R. Sadeghi is with the Department of Mathematics and Computer Science, Amirkabir University of Technology and F. Amirzade is with the Department of Mathematics, Shahrood University of Technology

(e-mail:  msadeghi@aut.ac.ir, famirzade@gmail.com).}
 \thanks{%
 Digital Object Identifier ????/TCOMM.?????}}


\maketitle
\begin{abstract}
In this paper we give lower bounds on the size of $(a,b)$ elementary trapping sets (ETSs) belonging to variable-regular LDPC codes with any girth, $g$, and irregular ones with girth 8, where $a$ is the size, $b$ is the number of degree-one check nodes and satisfy the inequality $\frac{b}{a}<1$. Our proposed lower bounds are analytical, rather than exhaustive search-based, and based on graph theories. The numerical results in the literarture for $g=6,8$ for variable-regular LDPC codes match our results. Some of our investigations are independent of the girth and rely on the variables $a$, $b$ and $\gamma$, the column weight value, only. We prove that for an ETS belonging to a variable-regular LDPC code with girth 8 we have $a\geq2\gamma-1$ and $b\geq\gamma$. We demonstrate that these lower bounds are tight, making use of them we provide a method to achieve the minimum size of ETSs belonging to irregular LDPC codes with girth 8 specially those whose column weight values are a subset of $\{2,3,4,5,6\}$. Moreover,  we show for variable-regular LDPC codes with girth 10,  $a\geq(\gamma-1)^2+1$. And for  $\gamma=3,4$ we obtain $a\geq7$ and $a\geq12$, respectively. Finally, for variable-regular LDPC codes with girths $g=2(2k+1)$ and $g=2(2k+2)$ we obtain  $a\geq(\gamma-2)^k+1$ and $a\geq2(\gamma-2)^k+1$, respectively. 

  \end{abstract}

%
\IEEEpeerreviewmaketitle
\section{Introduction}
\IEEEPARstart{O}  ne of the phenomena that influences significantly  the performance of  binary low-density parity-check codes (or simply LDPC codes) is known as $trapping$ $sets$. An $(a,b)$ trapping set of size $a$ is an induced subgraph of the Tanner graph on $a$ variable nodes and $b$ check nodes of odd degrees. Among all trapping sets, the most impressive ones are those with check nodes of degrees one or two. This category is so-called elementary trapping sets (or simply ETSs). In addition, according to the literature, ETSs that cause high decoding failure rate and are the main culprit of high error floor are those which  satisfy the inequality $\frac{b}{a}<1$. 

Hereafter, the symbol $\gamma$ and $\lambda$ stand for the column weight and the row weight of variable-regular LDPC codes, respectively. In  $\cite{2011}$ it is proved that a binary $(\gamma,\lambda)$-regular LDPC code whose Tanner graph is 4-cycle free contains no $(a,b)$ trapping set of size $a\leq \gamma$, where $\frac{b}{a}<1$. In $\cite{2014}$ Karimi et al. present a class of ETSs based on short cycles in the Tanner graph of variable-regular LDPC codes. Any member of a class such as $S$ is a sequence of ETSs starting from a short cycle which grows by one variable node at a time to generate $S$. Such obtained sequence is called a layered superset (LSS) of a short cycle. For  column weight values of $3,4,5,6$, girths $6,8$ and $a,b\leq10$ all classes of ETSs are presented. In $\cite{2015}$, LSSs of some graphical structures other than short cycles are demonstrated, by Hashemi et al., to find all $(a,b)$ elementary trapping sets. In $\cite{2016}$, they propose three expansion techniques to obtain all ETSs, referred to as leafless ETSs (LETSs), in which each variable node is connected to at least two even degree check nodes. This new characterization has some advantages compared to their counterparts. For example, it covers all the LETSs with  $a\leq a_{max}$ and $b\leq b_{max}$, for any value of $a_{max}$ and $b_{max}$. Moreover, short cycles enumerated have less lengths compared to LSS-based search. Furthermore, in $\cite{2010}$ by assuming $a\leq8$ and $\frac{b}{a}<1$, a characterization of $(a,b)$ trapping sets of $(3,\lambda)$-regular LDPC codes from Steiner triple systems is studied. Moreover, in $\cite{Vasic}$, a method to construct $(3,\lambda)$-regular LDPC code whose Tanner graph has girth 8 and contains a minimum number of small trapping sets is provided.  

In this paper, we  provide analytically lower bounds on the size of ETSs of a variable-regular LDPC code. Some of our results are independent of the girth of the Tanner graph. For example, we prove that if $\frac{b}{a}<1$ then a variable-regular LDPC code contains no $(a,b)$ ETS of size $a\leq\gamma$. Some others are provided directly according to the girth. We demonstrate that the Tanner graph of a variable-regular LDPC code with girth 8 contains no $(a,b)$ ETS of size $a\leq2\gamma-2$ and the minimum number of unsatisfied check nodes in ETSs is $\gamma$. We present a construction for ETSs belonging to variable-regular LDPC codes with girth 8  which shows the obtained lower bounds for $a$ and $b$ are tight for all values of $\gamma$.  These lower bounds on the size and the number of degree-one check nodes of ETSs provide us with a chance to present a method to achieve the lower bounds for the size of ETSs belonging to irregular LDPC codes. We apply this method on irregular LDPC codes  whose column weight values are a subset of $\{2,3,4,5,6\}$. 

Moreover, we show that the minimum size of $(a,b)$ ETSs related to a variable-regular LDPC code with girth 10 and $\gamma=3,4$ are 7 and 12, respectively. And we prove that, generally, variable-regular LDPC codes with girth 10 contain no ETS of size  $a\leq(\gamma-1)^2$. Finally, we generalize our results for   variable-regular LDPC codes with any girth and column weight values, as follows. Variable-regular LDPC codes with girths $g=2(2k+1)$ and $g=2(2k+2)$ contain no $(a,b)$ ETSs of sizes  $a\leq(\gamma-2)^k$ and $a\leq2(\gamma-2)^k$, respectively. We believe that the main contribution of this article is that if all types of $(a,b)$ ETSs of a Tanner graph with the property of $\frac{b}{a}<1$ are under consideration, the achieved lower bound can be the minimum value of the parameter $a$ as the initial input of search algorithms. 

The rest of the paper is organized as follows. Section II presents some basic notations, definitions and graph theories which are our principle tools to prove our results. In section III, we provide lower  bounds for the size of ETSs which are independent of the girth. Section IV presents the lower  bounds for the size of ETSs belonging to variable-regular and irregular LDPC codes with girth 8 as well as tables to illustrate the influence of lower bounds on investigating the existence and non-existence of some ETSs. Lower bounds for the size of ETSs of variable-regular LDPC codes with girth 10 are investigated in Section V. In section VI we generalize our results for variable-regulare LDPC codes with any girth.  In the last section we summarize our results.  

\section{Preliminaries}\label{}
One of the most important representations of codes is a Tanner graph. A Tanner graph is a bipartite graph in which the set of variable nodes (VNs) forms one of the vertex sets, which is denoted by $V$, and the set of check nodes (CNs) forms the other vertex set, which is denoted by $C$. The degree of a node $v$, either variable or check node, is denoted by $d(v)$. The minimum degree of nodes is denoted by $\delta$. The set of vertices connected to a vertex $v$ forms a neighbor set of $v$ and is denoted by $N(v)$. A variable-regular LDPC code is corresponding to a Tanner graph in which for all $v\in V$ we have $d(v)=\gamma$. Since a bipartite graph has no odd cycle, any cycle is represented by alternating sequence of check nodes and variable nodes. The length of the shortest cycle is called girth and denoted by $g$.

Take an induced subgraph of the Tanner graph on a subset $S$ of $V$. The subgraph contains some check nodes of odd degrees and some check nodes of even degrees referred to as  unsatisfied check nodes and satisfied check nodes, respectively. If $|S|=a$ and the number of unsatisfied check nodes is $b$ then the induced subgraph provides an $(a,b)$ trapping set of size $a$. An $(a,b)$ trapping set is called elementary if all check nodes are of degrees one or two. As a result, all unsatisfied check nodes in an elementary trapping set (or ETS) are of degree one. In this paper we concentrate on all ETSs whose parameters satisfy the inequality $\frac{b}{a}<1$ because of their significant influence on the error floor region.

In an $(a,b)$ ETS, by removing all unsatisfied check nodes and replacing every satisfied check node with an edge we obtain a graph with $a$ vertices which is called a normal graph. In this paper instead of searching all ETSs, we consider their corresponding normal graphs. 

In this section, we also provide some graph theories and definitions  which are our principle tools to prove our results. 
\begin{Definition}
Given a graph $G=(V,E)$, where $|V|$ and $|E|$ are the number of vertices and edges, respectively. The degree sum formula is as follows 
\begin{center}
$\sum_{v\in V}d(v)=2|E|$.
\end{center}
\end{Definition}
\begin{Definition}
A complete bipartite graph is a type of bipartite graph in which every vertex of the first vertex set is connected to every vertex of the second vertex set. A complete bipartite graph with $m$ vertices in one of the vertex sets and $n$ vertices in the other is denoted by $K_{m,n}$.   
\end{Definition}
\begin{Definition}
A complete  graph is a graph in which every pair of distinct vertices is connected by a unique edge. A complete  graph on $n$ vertices is denoted by $K_{n}$.   
\end{Definition}
\begin{Definition}
A triangle-free graph is a graph in which no three vertices form a cycle of length three. In other words, a triangle-free graph is a $K_3$-free graph.  
\end{Definition}

The existence of a $2k$-cycle in the ETS is equivalent to the existence of a cycle of length $k$ in its correspondent normal graph. For example, if a sequence of $v_0,c_0,v_1,c_1,v_2,c_2$ is a 6-cycle of an ETS, where $v_i\in V$ and $c_i\in C$, then by replacing any check node with an edge we obtain a cycle of length three in the normal graph whose vertices are $v_1,v_2,v_3$ . For example, in Fig. 1 $(a)$ the trapping set contains two 6-cycles and its corresponding normal graph has two triangles. Moreover, any 4-cycle in the trapping set is equivalent to a multiple edge in its corresponding normal graph. As an example, if a sequence of $v_0,c_0,v_1,c_1$ is a 4-cycle of a trapping set then by replacing any check node with an edge we have a multiple edge $(v_1,v_2)$.

In a 4-cycle free Tanner graph, the normal graph of each ETS is free of multiple edges which is called a simple graph. And in a Tanner graph with girth at least 8, the normal graph of each ETS is a simple and triangle-free graph. For example, in Fig. 1 $(b)$ the ETS contains no 6-cycle and its corresponding normal graph is triangle-free. 
\begin{center}
\begin{figure}
\centering
\includegraphics[scale=.3]{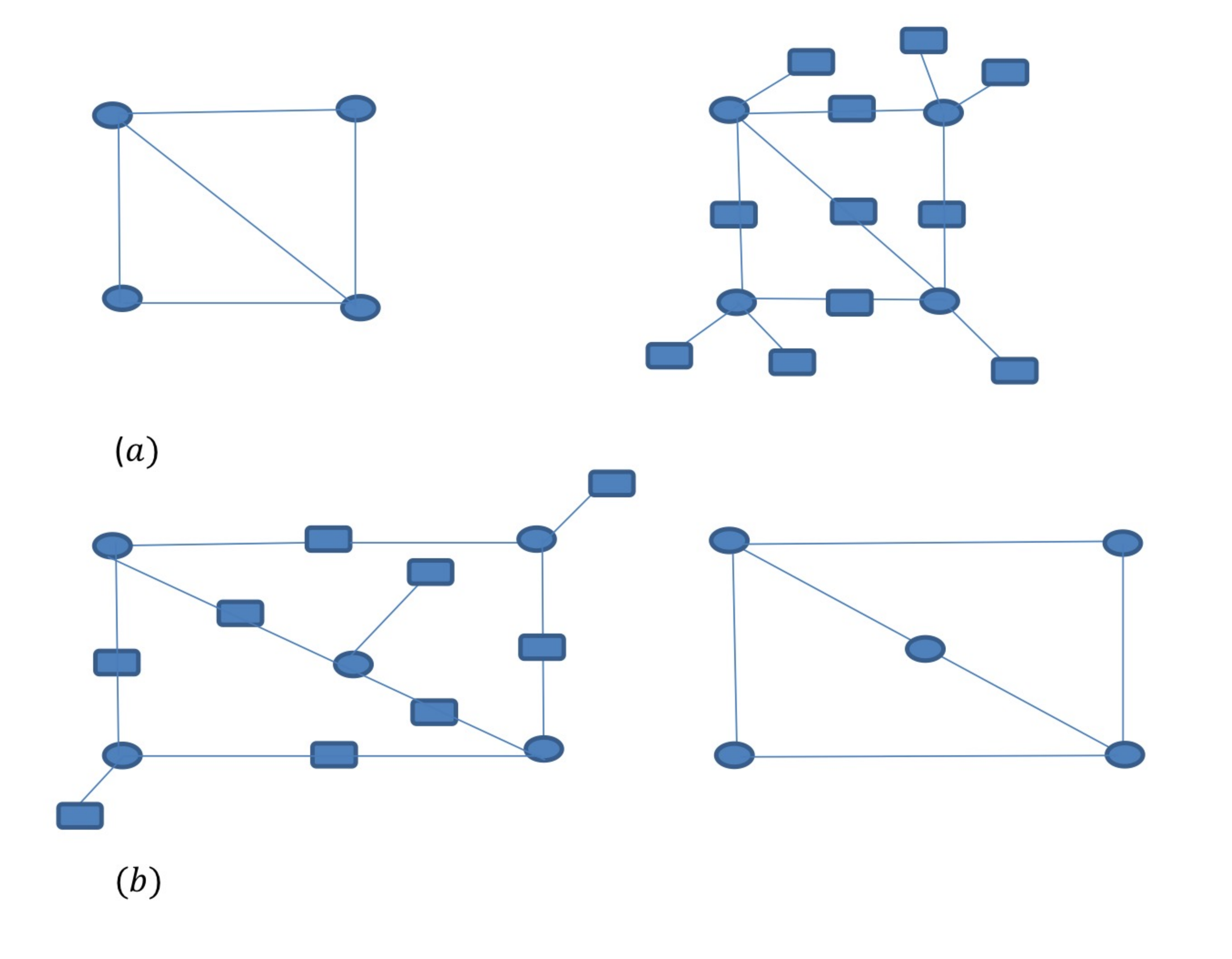}
\caption{$(a)$ is a normal graph and its corresponding (4,6) ETS with $\gamma=4$ and $(b)$ is a (5,3) ETS with $\gamma=4$ and its  corresponding normal graph}
\end{figure}
\end{center}
\section{Lower bounds for the size of ETSs which are independent of the girth}\label{}

In this section all the results are analytically obtained and are compared with those based on exhaustive search algorithms. First, we demonstrate some results regarding to the variables $a$, $b$ and $\gamma$ in an $(a,b)$ ETS. 

\begin{lemma} 
An $(a,b)$ ETS and its corresponding normal graph of a variable-regular Tanner graph, for which $\frac{b}{a}<1$, hold the following conditions.
	
\noindent $(i)$ The normal graph has at least one vertex of degree $\gamma$.
 
\noindent $(ii)$ There is no $(a,b)$ ETS, where $a\leq \gamma$.

\noindent $(iii)$ If the normal graph contains $|E|$ edges then $b=a\gamma-2|E|$. 

\noindent $(iv)$ If $a$ is an even number then $b$ is also an even number.

\noindent $(v)$ If $a$ is an odd number then parameters $\gamma$ and $b$ both are even or odd.
\end{lemma}
\begin{proof}
$(i)$ Since our desire $(a,b)$ ETSs are those hold the inequality $\frac{b}{a}<1$,  the normal graph has at least a  vertex of degree $\gamma$. Because, if for each vertex, $v$, we have $d(v)<\gamma$ then $\gamma-d(v)\geq1$. Therefore $b=\sum_{i=1}^{a}(\gamma-d(v_i))\geq a$ and so $\frac{b}{a}\geq1$. 

 $(ii)$  It is clear that in any graph with $a$ vertices, the degree of each vertex is at most $a-1$, since the normal graph has at least a  vertex of degree $\gamma$ we have $\gamma\leq a-1$.
	
$(iii)$ According to the degree sum formula we have: $b=\sum_{i=1}^{a}(\gamma-d(v_i))=a\gamma-\sum_{i=1}^{a}d(v_i)=a\gamma-2|E|$.
	
 $(iv)$ Since $a\gamma=2|E|+b$, if $a$ is an even number then the left side of the equality is even so $b$ is also an even number.
	
 $(v)$  If $a$ is an odd number then the parameter $\gamma$ determines $b$ as an odd or even number. In this case if $\gamma$ is odd then the left side of the equality, $a\gamma=2|E|+b$, is also an odd number, as a result the right side must be an odd number which proves $b$ is odd. If $\gamma$ is even then the left side of the equality is also an even number, as a result the right side must be an even number which proves $b$ is even.	
\end{proof}

An immediate consequence of the above Lemma is shown in Table I for $\gamma=3,4,5,6$  which demonstrates the non-existence of some of $(a,b)$ ETSs of a variable-regular LDPC code. The table indicates that in order to consider the existence of some $(a,b)$ ETSs there is no need to apply the proposed exhaustive search algorithms in the literature. Moreover, it expresses that the non-existence of some of $(a,b)$ ETSs is independent of the girth of the Tanner graph, while in \cite{2014} they are obtained by exhaustive search algorithms for Tanner graphs with girth 6. Additionally, the sign of "-" in the table illustrates the non-existence of corresponding $(a,b)$ ETS which is proved by Lemma 1 $(ii)$.

\begin{table}[h]
\begin{center}
\begin{tabular}{|c|c|c|c|c|}
\hline
    & $\gamma=3$ & $\gamma=4$ & $\gamma=5$ & $\gamma=6$\\ 
\hline    
$a=4$ & $b=1,3$ & -  & - & -\\
\hline
$a=5$ & $b=0,2,4$ & $b=1,3$ & - & -\\
\hline
$a=6$ & $b=1,3,5$ & $b=1,3,5$ & $b=1,3,5$ & -\\
\hline
$a=7$ & $b=0,2,4,6$ & $b=1,3,5$ & $b=0,2,4,6$ & $b=1,3,5$\\
\hline
$a=8$ & $b=1,3,5,7$ & $b=1,3,5,7$ & $b=1,3,5,7$ & $b=1,3,5,7$\\
\hline
$a=9$ & $b=0,2,4,6,8$ & $b=1,3,5,7$ & $b=0,2,4,6,8$ & $b=1,3,5,7$ \\
\hline
\end{tabular}
\caption{Non-existence of $(a,b)$ ETSs of variable-regular LDPC codes}
\end{center}
\end{table}

\section{Lower bounds on the size of ETSs for both variable-regular and irregular LDPC codes with girth 8}\label{}
In this section and the following ones we provide lower bounds on the size of elementary trapping sets according to the girth of the Tanner graph. Although in the section VI we obtain a lower bound for the size of ETSs belonging to Tanner graphs with any girth $g\geq8$, the lower bounds for the size of ETSs belonging to Tanner graphs with girths 8 and 10 are investigated separately to make  the  lower bounds tighter.

This section contains two parts. In the first part, we consider the size  and the number of degree-one check nodes of ETSs belonging to variable-regular LDPC codes whose Tanner graphs have girth 8. In the second part, we take benefit from the results in the first part and present a method to determine a lower bound for the size of ETSs belonging to irregular LDPC codes whose Tanner graphs have girth 8. We provide numerical results for irregular LDPC codes whose variable nodes are a subset of $\{2,3,4,5,6\}$. 
\subsection{ Variable-regular LDPC codes with girth 8}
In \cite{2014} there are some numerical results, which are search-based, about ETSs belonging to Tanner graphs with girth 8 and $\gamma=3,4,5$. In order to consider the lower bound for the size of ETSs belonging to LDPC codes with girth 8 and any value of $\gamma$ we provide the well-known Turan's Theorem about all $K_{r+1}$-free graphs, as follows.
\begin{Theorem}~\label{lemOrder}
	$\cite{Turan}$ Let $G$ be any graph with $n$ vertices, such that $G$ is $K_{r+1}$-free. Then the number of edges in $G$ is at most $\frac{r-1}{r}\times\frac{n^2}{2}=(1-\frac{1}{r})\times\frac{n^2}{2}$.
\end{Theorem} 
As mentioned in Section II, the normal graph corresponding to an ETS belonging to LDPC codes with girth 8 is a $K_3$-free graph. The following proposition for triangle-free graphs, which is a consequence of the above theorem, is our main tool to obtain the lower bounds for $a$ and $b$.

\begin{Proposition}
Let $G$ be a triangle-free graph with $n$ vertices, then the number of edges in $G$ is at most $\frac{n^2}{4}$.
\end{Proposition} 
\begin{Theorem}~\label{lemOrder}
In a variable-regular LDPC code, whose Tanner graph has girth  at least 8, there is no $(a,b)$ ETS of size less than $2\gamma-1$, where $a,b$  satisfy the inequality $\frac{b}{a}<1$.
\end{Theorem} 
\begin{proof}
According to Lemma 1 $(iii)$ we have $b=a\gamma-2|E|$, where $|E|$ is the number of edges in the normal graph corresponding to an $(a,b)$ ETS. Since the Tanner graph has girth at least 8, it is 6-cycle free. Therefore as mentioned in the previous section, the normal graph is triangle-free and so according to proposition 1 the maximum number of edges is $\frac{a^2}{4}$. Assume that the normal graph has the maximum number of edges. In this case $b=a\gamma-2|E|=a\gamma-2(\frac{a^2}{4})$. So the number of unsatisfied check nodes is $a\gamma-\frac{a^2}{2}$. In order to consider dominant $(a,b)$ ETSs we focus on those  with the property of $\frac{b}{a}<1$. So we have $\frac{b}{a}=\frac{a\gamma-\frac{a^2}{2}}{a}=\gamma-\frac{a}{2}<1$, which results in the following inequality $2\gamma-2<a$. As a consequence, there is no $(a,b)$ ETS of size less than or equal to $2\gamma-2$ with the property of $\frac{b}{a}<1$.
\end{proof}
\begin{Theorem}
For each $(a,b)$ ETS of a variable-regular LDPC code, whose Tanner graph has girth at least 8, the number of odd degree check nodes is at least $a\gamma-\frac{a^2}{2}$.
\end{Theorem}
\begin{proof}
As considered in Lemma 1 $(iii)$, if the number of edges in a normal graph is $|E|$ then $a,b,\gamma$ and $|E|$ hold the equality $b=a\gamma-2|E|$. If the Tanner graph has girth at least 8 then $2|E|\leq \frac{a^2}{2}$. Therefore $b=a\gamma-2|E|\geq a\gamma-\frac{a^2}{2}$.  
\end{proof}

For example, if an ETS contains the minimum size,  $a=2\gamma-1$, then $b=\gamma$. In the following we explain how a $(2\gamma-1,\gamma)$ ETS belonging to a variable-regular LDPC code is constructed. This structure proves that the obtained lower bounds for $a$ and $b$ are tight. In this case  the corresponding normal graph has $a=2\gamma-1$ vertices and so has the maximum number of edges, which is $|E|=[\frac{(2\gamma-1)^2}{4}]=\gamma^2-\gamma$. Such a normal graph can be a complete bipartite graph $K_{\gamma-1,\gamma}$. Therefore, normal graphs corresponding to ETSs where $\gamma=3,\gamma=4,\gamma=5$ and $\gamma=6$ are $K_{2,3},K_{3,4},K_{4,5}$ and $K_{5,6}$, respectively. The corresponding ETSs are shown in Fig. 2.
\begin{center}
\begin{figure}[h!]
\centering
\includegraphics[scale=.3]{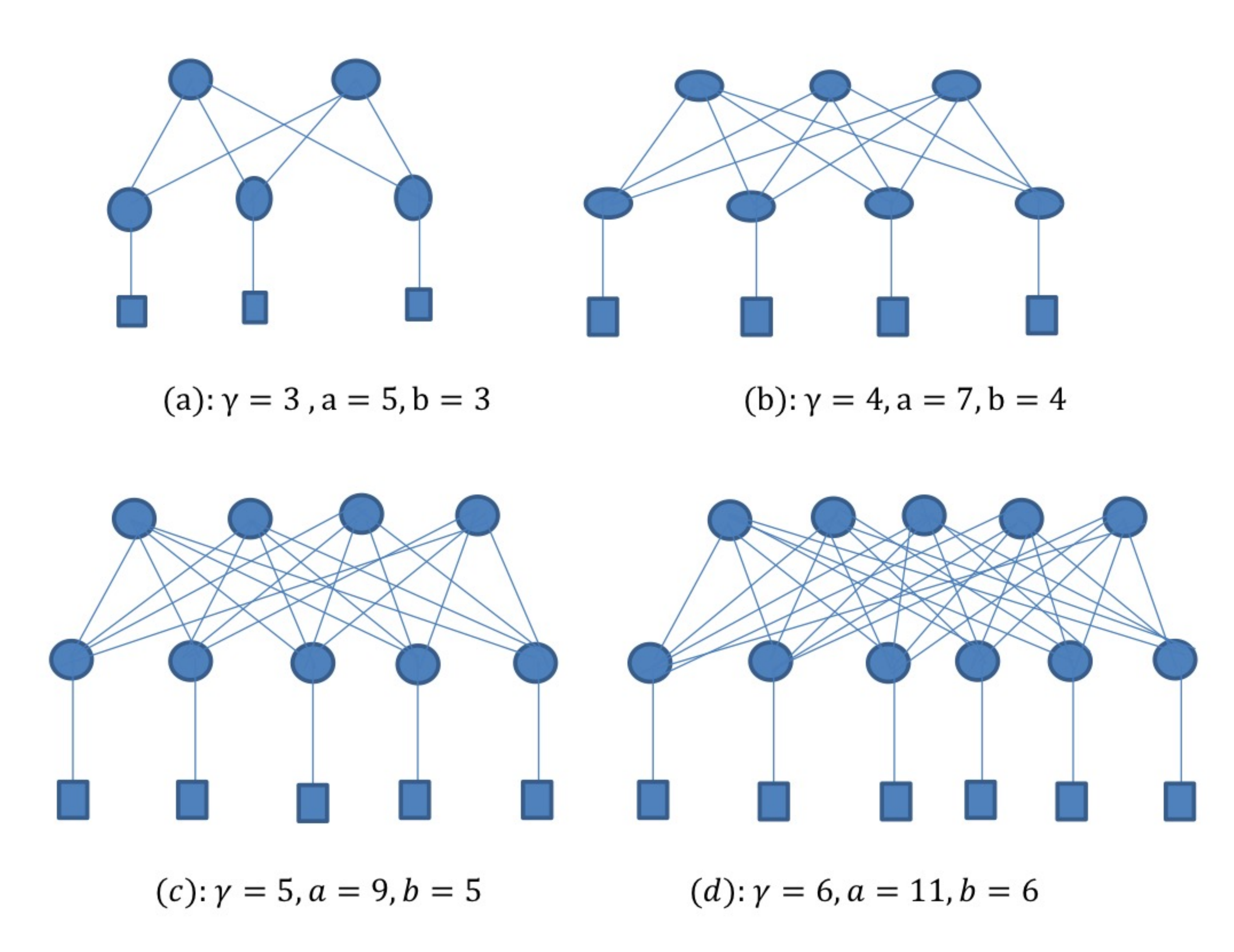}\\
\caption{ETSs for $\gamma=3,\gamma=4,\gamma=5$ and $\gamma=6$}
\end{figure}
\end{center}

The lower bounds obtained for $a$ and $b$ in Theorem 2 and Theorem 3 as well as Lemma 1 provide us with useful information about the existence of ETSs without using any search-based method. As an example, for variable-regular LDPC codes whose Tanner graphs have $\gamma=6$ and $g=8$ we have $a\geq11$ and $b\geq 6a-\frac{a^2}{2}$. So for $a=11$  the minimum number of unsatisfied check nodes is 6. And finally by applying Lemma 1 we conclude that there are three ETSs to investigate, which are (11,6) ETS, (11,8) ETS and (11,10) ETS. 

For $\gamma=3,4,5$ we provide Table II to demonstrate the existence of some $(a,b)$ ETSs with the property of $\frac{b}{a}<1$ belonging to variable-regular LDPC codes whose Tanner graph have $g=8$. The table also expresses that the non-existence of some of $(a,b)$ ETSs is independent of the girth of the Tanner graph and relies on $a,b$ and $\gamma$. Note that the lower bound we provided analytically is used for any given value of $\gamma$.
  
\begin{table}[h]
\begin{center}
\begin{tabular}{|c|c|c|c|c|}
\hline
    & $\gamma=3$ & $\gamma=4$ & $\gamma=5$\\ 
\hline    
$a=4$ &  &   & \\
\hline
$a=5$ & $b=3$ &  & \\
\hline
$a=6$ & $b=0,2,4$ &  & \\
\hline
$a=7$ & $b=1,3,5$ & $b=4,6$ & \\
\hline
$a=8$ & $b=0,2,4,6$ & $b=0,2,4,6$ & \\
\hline
$a=9$ & $b=1,3,5,7$ & $b=0,2,4,6,8$ & $b=5,7$\\
\hline
\end{tabular}
\caption{Existence of $(a,b)$ ETSs of variable-regular LDPC codes whose Tanner graph have $g=8$}
\end{center}
\end{table}
\subsection{Irregular LDPC codes with girth 8}
There are a number of irregular LDPC codes with different variable node degrees. Hence, considering all types of ETSs belonging to irregular LDPC codes seems difficult and rather impossible. Although there is a search-based algorithm in \cite{2012} to find ETSs belonging to a given irregular LDPC code, it is not comprehensive. And, as a whole, there are not much works done to determinate all types of ETSs in this catagory. However, the results proposed in the part A provide us with a chance to find out a method to obtain lower bounds for the size of ETSs belonging to irregular LDPC codes whose Tanner graphs have girth 8. We apply the method on irregular LDPC codes whose variable node degrees belong to the set $\{2,3,4,5,6\}$. 

In this case our desire $(a,b)$ ETSs are also those satisfying in the inequality $\frac{b}{a}<1$. From Theorem 2 we conclude that if the minimum degree of variable nodes is $\delta$ then we have $a\geq2\delta-1$. In order to consider the lower bound for the size of ETSs belonging to an irregular LDPC code we first take into account an  ETS belonging to a variable-regular LDPC code with $\gamma=\delta$, $a=2\delta-1$ and $b=\delta$ like those shown in Fig. 2. Depending on the difference between two parameters $a$ and $b$ there are two cases to consider.
\begin{itemize}
\item If $a-b=1$ then we have to increase the number of variable nodes in the ETSs belonging to a variable-regular LDPC code with $\gamma=\delta$ at least by one. Because, adding a check node of degree 1 to an ETS belonging to a variable-regular LDPC code with $\gamma=\delta$ results in an $(a,a)$ ETS belonging to an irregular LDPC code whose column weight values contain $\delta$ and $\delta+1$. But for this case we have $\frac{b}{a}=1$ which contradicts our desire.  

For example, suppose the lower bound for the size of ETSs belonging to irregular LDPC codes whose column weights are  in a set $\{2,3\}$ is under consideration. For this case we have $\delta=2$ and according to Theorem 2 we have $a\geq3$. We first take a $(3,2)$ ETS belonging to a variable-regular LDPC code with $\gamma=2$. Since $a-b=1$ inorder to obtain an  ETS belonging to irregular LDPC codes whose column weights are  in a set $\{2,3\}$ we have to increase the number of variable nodes of the mentioned $(3,2)$ ETS at least by one. Because by adding a check node of degree 1 to the mentioned $(3,2)$ ETS we obtain a $(3,3)$ ETS for which we have $\frac{b}{a}=1$. In this case $a=4$ is considered as the lower bound for the size of ETSs belonging to an irregular LDPC code with $d(v)\in\{2,3\}$. In Fig. 3 we present all types of ETSs with the minimum size 4 belonging to four types of irregular LDPC codes. They are a (4,1) ETS for $d(v)\in\{2,3\}$, a (4,2) ETS for $d(v)\in\{2,4\}$, two (4,3) ETSs for $d(v)\in\{2,5\}$  and  $d(v)\in\{2,3,4\}$.
\begin{center}
\begin{figure}[h!]
\centering
\includegraphics[scale=.3]{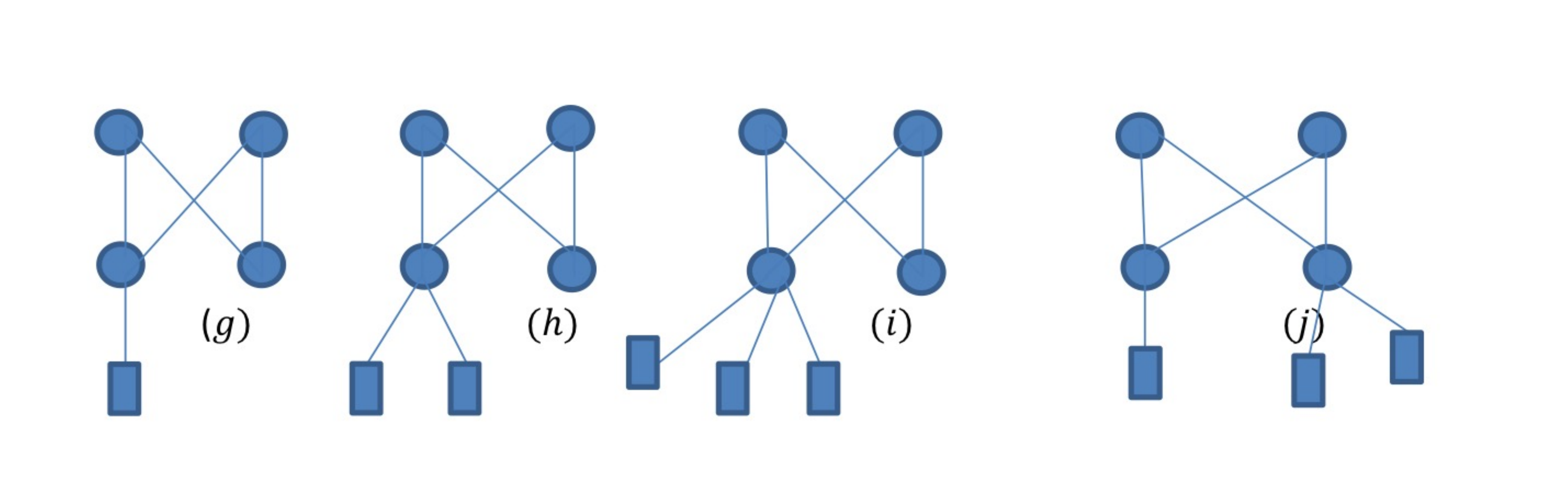}\\
\caption{$(g)$ is a (4,1) ETS with $d(v)\in\{2,3\}$, $(h)$ is a (4,2) ETS with $d(v)\in\{2,4\}$, $(i)$ is a (4,3) ETS with $d(v)\in\{2,5\}$ and  $(j)$ is a (4,3) ETS with $d(v)\in\{2,3,4\}$}
\end{figure}
\end{center}
\item If $a-b=i\geq2$ then the lower bound for the size of the ETS obtained for variable-regular LDPC codes with the column weight $\delta$ can be accounted for irregular LDPC codes with the minimum degree $\delta$, as well. On the condition that the number of degree-one check nodes which have to be added to a $(2\delta-1,\delta)$ ETS to achieve an ETS belonging to irregular LDPC codes with minimum degree $\delta$ is less than $i$, otherwise in order to obtain an ETS belonging to an irregular LDPC code from an ETS belonging to a variable-regular LDPC code with $\gamma=\delta$, the number of variable nodes  has to be raised at least by one.

For example, suppose the lower bounds for the size of ETSs belonging to irregular LDPC codes whose column weights are  in a set $\{4,5\}$, $\{4,6\}$ or $\{4,7\}$ are under consideration.  For all three cases we have $\delta=4$. According to Theorem 2, for a variable-regular LDPC code whose column weight is 4 we have $a\geq7$. We first take a (7,4) ETS belonging to a variable-regular LDPC code with $\gamma=4$. Since $a-b=3$, the lower bound for the size of ETSs belonging to irregular LDPC codes whose column weights are  in a set $\{4,5\}$ or $\{4,6\}$ is also 7, because in these cases the number of degree-one check nodes added to the (7,4) ETS is at most 2. But if column weights are in the set $\{4,7\}$ then the size of ETSs is at least 8. To clarify this construction we present Fig. 4. 
\begin{center}
\begin{figure}[h!]
\centering
\includegraphics[scale=.3]{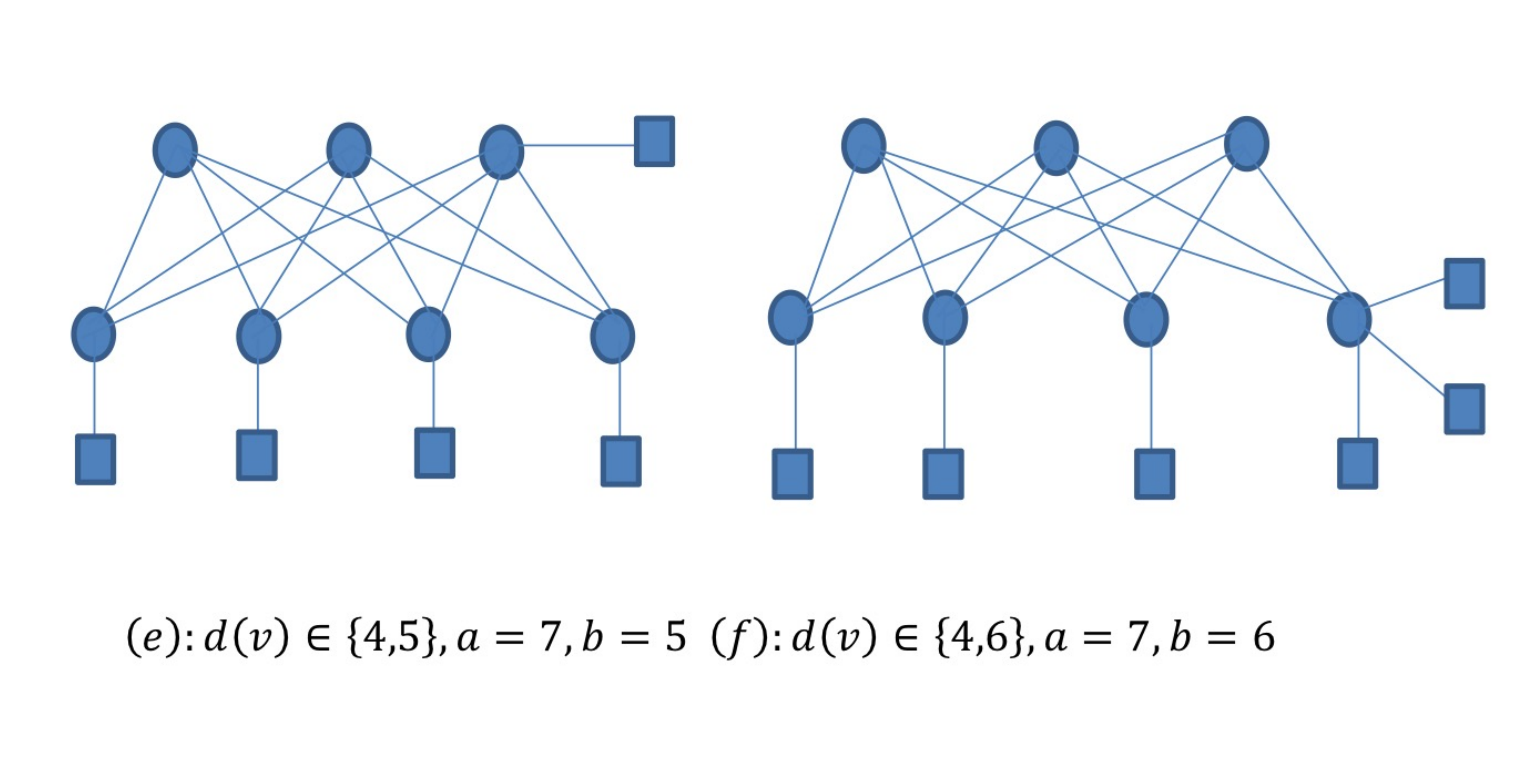}\\
\caption{$(e)$ and $(f)$ are a (7,5) ETS and  a (7,6) ETS obtained from (7,4) ETS shown in Fig. 2 $(c)$}
\end{figure}
\end{center} 
\end{itemize}
Generally, we obtain the exact lower bound for the desire irregular LDPC codes by the use of the difference between two parameters $a$ and $b$. In the following table we provide the lower bounds for ETSs belonging to irregular LDPC codes whose column weights are a subset of $\{2,3,4,5,6\}$ and their Tanner graphs have girth 8.   
\begin{table}[h]
\begin{center}
\begin{tabular}{|c|c|}
\hline
$a=4$ & (4,1) ETS and $d(v)\in\{2,3\}$, (4,2) ETS and $d(v)\in\{2,4\}$\\ 
    & (4,3) ETS and $d(v)\in\{2,5\}$, (4,3) ETS and $d(v)\in\{2,3,4\}$\\ 
\hline    
$a=5$ & (5,4) ETS and $d(v)\in\{2,6\}$, (5,4) ETS and $d(v)\in\{3,4\}$\\
     & (5,2) ETS and $d(v)\in\{2,3,5\}$, (5,3) ETS and $d(v)\in\{2,3,6\}$\\
     & (5,3) ETS and $d(v)\in\{2,4,5\}$, (5,4) ETS and $d(v)\in\{2,4,6\}$\\
     & (5,4) ETS and $d(v)\in\{2,3,4,5\}$\\
\hline
$a=6$ & (6,2) ETS and $d(v)\in\{3,5\}$, (6,3) ETS and $d(v)\in\{3,6\}$\\
    & (6,3) ETS and  $d(v)\in\{2,5,6\}$, (6,3) ETS and $d(v)\in\{3,4,5\}$\\
    & (6,4) ETS and $d(v)\in\{3,4,6\}$, (6,5) ETS and $d(v)\in\{3,5,6\}$\\
    & (6,5) ETS and  $d(v)\in\{2,4,5,6\}$, (6,4) ETS and  $d(v)\in\{2,3,5,6\}$\\
\hline
$a=7$ & (7,5) ETS and $d(v)\in\{4,5\}$, (7,6) ETS and $d(v)\in\{4,6\}$\\ 
   & (7,3) ETS and $d(v)\in\{3,4,5,6\}$, (7,4) ETS and $d(v)\in\{2,3,4,5,6\}$ \\
\hline
$a=8$ & $d(v)\in\{4,6\}$, $d(v)\in\{4,5,6\}$\\
\hline
$a=9$ & $d(v)\in\{5,6\}$\\
\hline
\end{tabular}
\caption{Lower bounds for the size of ETSs belonging to irregular LDPC codes whose column weights are a subset of $\{2,3,4,5,6\}$ and Tanner graphs have girth 8}
\end{center}
\end{table} 

\section {Lower bounds for the size of ETSs for variable-regular LDPC codes with girth 10}
The normal graph corresponding to an elementary trapping set belonging to LDPC codes with girth 10 is triangle-free and has no 4-cycles. The following definition and theorem contribute to present lower bounds of the size of ETSs.
\begin{Definition}
	Suppose a cycle of length $i$ is denoted by $C_i$. If a simple graph has girth $g$ then it is $i$-cycle free for each $3\leq i\leq g-1$. The maximum number of edges of a graph with $n$ vertices and girth $g$ is denoted by $ex(n,{C_3,C_4,\dots,C_{g-1}})$. 
\end{Definition}
For example, in a triangle-free graph we have: $ex(n,C_3)\leq\frac{n^2}{4}$. To consider the maximum number of edges of a graph with girth  5 we utilize the following  Theorem.
\begin{Theorem}~\label{lemOrder}
	$\cite{Garnick}$ For a graph with $n$ vertices and girth $g=5$, the maximum number of edges is  as follows:
	\begin{center}
		$ex(n,{C_3,C_4})\leq\frac{1}{2}n^\frac{3}{2}$.
	\end{center}
\end{Theorem}

Now, by assuming the maximum number of edges in a normal graph with girth 5 we propose the following theorem related to the Tanner graph with girth 10.
\begin{Theorem}~\label{lemOrder}
In a variable-regular LDPC code, whose Tanner graph has girth at least 10, there is no $(a,b)$ ETS of size less than $(\gamma-1)^2+1$, that is $a\geq (\gamma-1)^2+1$, where $a,b$  satisfy the inequality $\frac{b}{a}<1$.
\end{Theorem} 
\begin{proof}
Take an $(a,b)$ ETS whose normal graph has the maximum number of edges. In this case $b=a\gamma-2|E|=a\gamma-2(\frac{a^\frac{3}{2}}{2})$. So the number of unsatisfied check nodes is $a\gamma-a^\frac{3}{2}$. In order to consider dominant $(a,b)$ ETSs we focus on those ETSs with the property of $\frac{b}{a}<1$. So we have $\frac{b}{a}=\frac{a\gamma-a^\frac{3}{2}}{a}=\gamma-a^\frac{1}{2}<1$.  As a consequence, there is no $(a,b)$ ETS of size less than or equal to $(\gamma-1)^2$ with the property of $\frac{b}{a}<1$.
\end{proof}

In the following we provide some properties related to a normal graph corresponding to an ETS belonging to a Tanner graph with girth at least 10. By this consideration we conclude that the proposed lower bound above can be also improved by more investigations.
\begin{itemize} 
\item As mentioned in Lemma 1 $(i)$, the normal graph  has at least a vertex of degree $\gamma$. Suppose $d(v)=\gamma$. If there is an edge between two vertices of neighbours of $v$ then the normal graph has a triangle. So there are no edges between neighbours of $v$.  
\item Suppose $z,z'\in N(v)$, in other words $z$ and $z'$ are connected to $v$. If $w$ is a neighbour of $z$ then it does not connect $z'$. Otherwise, $z,z'\in N(w)$ and if $|N(v)\cap N(w)|\geq2$ then the normal graph has 4-cycles. As a result each vertex other than $v\cup N(v)$ has only one connection with $v\cup N(v)$.   
\end{itemize}
In the following we apply Theorem 5 as well as the two items above to provide the exact lower bound for the size of an ETS where $\gamma=3$ and $\gamma=4$.
\begin{itemize}
\item ETSs belonging to a Tanner graph with girth 10 and $\gamma=3$ have no triangles and is 4-cycle free so according to Theorem 5 it has the size of at least 5. By making use of the mentioned properties we conclude that $a\geq7$. A normal graph with 7 vertices and girth 5 is illustrated in Fig. 5 $(a)$.

\item We prove for $\gamma=4$ and $g=10$ the lower bound of the size of ETSs is 12. Theorem 5 gives $a\geq10$. According to Lemma 1 $(iv)$, if we take $a=10$ then we have $b\leq8$. Suppose that there is a $(10,8)$ ETS. Lemma 1 $(iii)$ proves that the number of edges in the normal graph is 16. According to Theorem 4 the  maximum number of edges of the graph with girth 5 and  10 vertices is 15. So the normal graph with 16 edges has a triangle or a 4-cycle.  Now by similar proof we illustrate that for $\gamma=4$ and $g=10$ there is no $(11,b)$ ETS for which $\frac{b}{a}<1$. The maximum number of edges of a graph with 11 vertices and girth 5 ,which is obtained in $\cite{Garnick}$, is 16. Therefore, if we take $b=10$ then the maximum number of edges in the normal graph is 17 and so it has a triangle or a 4-cycle. An example of such a normal graph is shown in  Fig. 5 $(b)$ which has a 4-cycle.  Finally, Fig. 5 $(c)$ demonstrates a $(12,10)$ ETS with $\frac{b}{a}<1$.
\begin{figure}[h!]
\centering
\includegraphics[scale=.3]{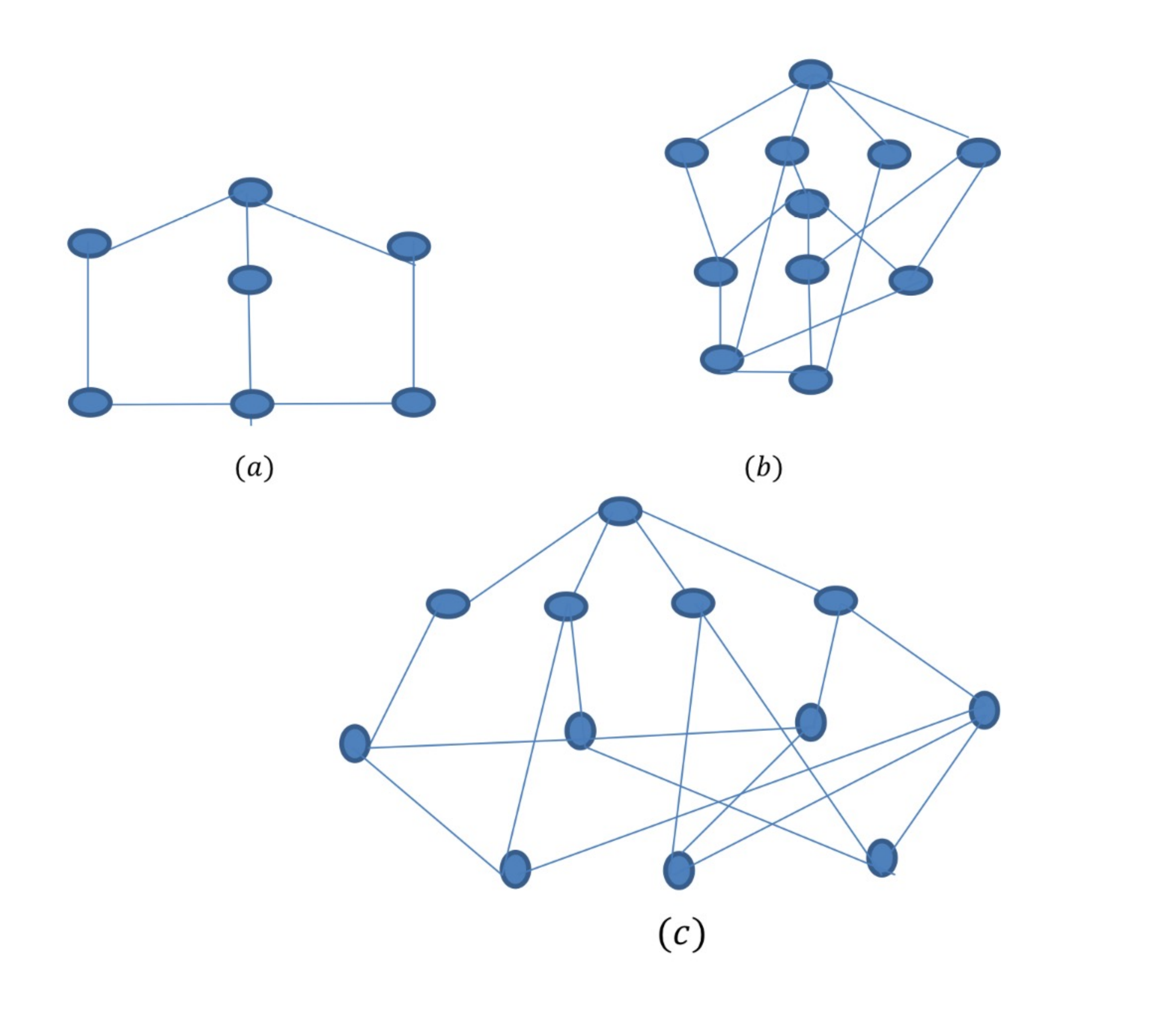}\\
\caption{$(a)$ is a normal graph corresponding to a (7,5) ETS with $\gamma=3$, $(b)$ is a normal graph corresponding to a (11,10) ETS with $\gamma=4$ which has a 4-cycle and $(c)$ is a normal graph corresponding to a (12,10) ETS with $\gamma=4$ which is free of triangles and 4-cycles}
\end{figure}
\end{itemize}
\section{Lower bounds for the size of ETSs for variable-regular LDPC codes with $g\geq12$}

Now we generalize our results for all values of the girth in two steps. In the first step we investigate the lower bound for the size of ETSs belonging to variable-regular LDPC codes whose Tanner graph have girth $g=2(2k+1)$. In this case the normal graph has girth $g=2k+1$ and according to Difinition 5 it is $i$-cycle free for each $3\leq i\leq 2k$ for which we have the following theorem as our main tool to obtain our result.
\begin{Theorem}~\label{lemOrder}
	$\cite{Furedi}$ For a graph with $n$ vertices and girth $g=2k+1$, the maximum number of edges is  as follows:
	\begin{center}
		$ex(n,{C_3,C_4,\dots,C_{2k}})<\frac{1}{2}n^{1+\frac{1}{k}}+\frac{1}{2}n$.
	\end{center}
\end{Theorem} 

\begin{Theorem}~\label{lemOrder}
In a variable-regular LDPC code, whose Tanner graph has girth $g=2(2k+1)$, there is no $(a,b)$ ETS of size less than or equal to $(\gamma-2)^k$, where $a,b$  satisfy the inequality $\frac{b}{a}<1$.
\end{Theorem} 
\begin{proof}
The normal graph has $a$ vertices and $|E|$ edges. Since the Tanner graph has girth $g$, it is $2i$-cycle free for all $3\leq i\leq \frac{g}{2}-1$. Therefore, according to Definition 5, the normal graph is $i$-cycle free and so according to Theorem 6 the maximum number of edges is less than  $\frac{1}{2}a^{1+\frac{1}{k}}+\frac{a}{2}$. Without loss of generality, assume that the  number of edges in the normal graph is $\frac{1}{2}a^{1+\frac{1}{k}}+\frac{a}{2}-1$. In this case $b=a\gamma-2|E|=a\gamma-2(\frac{1}{2}a^{1+\frac{1}{k}}+\frac{a}{2}-1)$. So we have $\frac{b}{a}=\frac{a\gamma-a^{1+\frac{1}{k}}-a+2}{a}=\gamma-a^{\frac{1}{k}}-1+\frac{2}{a}<1$. Therefore, the inequality $\gamma-a^{\frac{1}{k}}-1<\gamma-a^{\frac{1}{k}}-1+\frac{2}{a}<1$ results in the inequality $\gamma-2<a^{\frac{1}{k}}$. As a consequence, there is no $(a,b)$ ETS of size less than or equal to $(\gamma-2)^k$ with the property of $\frac{b}{a}<1$.
\end{proof}

In the second step we investigate the lower bound for the size of ETSs belonging to variable-regular LDPC codes whose Tanner graphs have girth $g=2(2k+2)$. In this case the normal graph has girth $g=2(k+1)$ and according to Difinition 5 it is $i$-cycle free for each $3\leq i\leq 2k+1$ for which we have the following theorem which contributes to achieve our result.
\begin{Theorem}~\label{lemOrder}
	$\cite{Furedi}$ For a graph with $n$ vertices and girth $g=2k+2$, the maximum number of edges is  as follows:
	\begin{center}
		$ex(n,{C_3,C_4,\dots,C_{2k+1}})<\frac{1}{2^{1+\frac{1}{k}}}n^{1+\frac{1}{k}}+\frac{1}{2}n$.
	\end{center}
\end{Theorem}

\begin{Theorem}~\label{lemOrder}
In a variable-regular LDPC code, whose Tanner graph has girth $g=2(2k+2)$, there is no $(a,b)$ ETS of size less than or equal to $2(\gamma-2)^k$, where $a,b$  satisfy the inequality $\frac{b}{a}<1$.
\end{Theorem} 
\begin{proof}
 The Tanner graph has girth $g$, it is $2i$-cycle free for all $3\leq i\leq\frac{g}{2}-1$. Therefore, the normal graph is $i$-cycle free and according to Theorem 8 the maximum number of edges is less than  $\frac{1}{2^{1+\frac{1}{k}}}a^{1+\frac{1}{k}}+\frac{a}{2}$. Assume that the normal graph has $\frac{1}{2^{1+\frac{1}{k}}}a^{1+\frac{1}{k}}+\frac{a}{2}-1$ edges. In this case $\frac{b}{a}=\frac{a\gamma-2(\frac{1}{1+2^{\frac{1}{k}}}a^{1+\frac{1}{k}}+\frac{a}{2}-1)}{a}$. By considering the inequality  $\frac{b}{a}<1$ and similar to the proof of Theorem 7 we obtain $\gamma-\frac{a}{2}^\frac{1}{k}<2$. Consequently, there is no $(a,b)$ ETS of size less than or equal to $2(\gamma-2)^k$ with the property of $\frac{b}{a}<1$.
\end{proof}
\section{Conclusion}\label{}
In this paper, we  provide lower bounds for the size of $(a,b)$ elementary trapping sets for variable-regular LDPC codes with any girth and irregular LDPC codes with girth 8, where $a$ is the number  of variable nodes and $b$ is the number of check nodes of odd degrees which satisfy the inequality $\frac{b}{a}<1$. We analytically demonstrate that depending on the number of variable nodes, the number of degree-one check nodes and the column weight, some of ETSs do not exist. It indicates that the non-existence of some of $(a,b)$ ETSs not only is independent of the girth of the Tanner graph but also in order to consider them  we do not need to conduct exhaustive search algorithms. Also, making use of some results on the maximum number of edges of a graph based on its girth, we  provide the lower bounds for the size of ETSs for Tanner graphs with girths 8, 10 and larger. In fact, we prove that the Tanner graph of a variable-regular LDPC code, with the column weight  $\gamma$ and girth 8 contains no $(a,b)$ ETS of size $a<2\gamma-1$. We also demonstrate that the lower bound of degree-one check nodes is $\gamma$. We show that these lower bound are tight. Along with making use of them we present a method to obtain  the lower bounds of ETSs belonging to irregular LDPC codes. We apply our proposed method on irregular LDPC codes whose column weight values are a subset of $\{2,3,4,5,6\}$. In addition, we prove that  variable-regular LDPC codes with girth 10 contain no ETS of size  $a\leq(\gamma-1)^2$. Moreover, for this case the lower bounds for $a$, assuming $\gamma=3$ and 4, are also improved which are 7 and 12, respectively.  Finally, we generalize our results for all values of the girth, as follows. Variable-regular LDPC codes with girths $g=2(2k+1)$ and $g=2(2k+2)$ contain no $(a,b)$ ETSs of sizes  $a\leq(\gamma-2)^k$ and $a\leq2(\gamma-2)^k$, respectively.

\end{document}